\documentclass[leqno,11pt]{amsart}
\usepackage{amsmath,amsfonts,amsthm,amssymb,indentfirst,epic,url,centernot}
\usepackage{color}
\usepackage{graphicx}
\usepackage{caption}
\usepackage[utopia]{mathdesign}
\usepackage{algorithm}
\usepackage[noend]{algorithmic} 
\algsetup{indent=3mm} 

\newtheorem{Theorem}{Theorem}
\newtheorem{Example}{Example}[section]
\newtheorem{Remark}{Remark}
\newtheorem{Notation}{Notation}

\newtheorem{Definition}{Definition}[section]
\newtheorem{Lemma}{Lemma}

\newtheorem{Proposition}{Proposition}
\newtheorem{Corollary}{Corollary}
\newenvironment{Proof}[1][Proof]{\noindent {\bf Proof.}}{\qed}

\begin{document}

\title{Contraction of Cyclic Codes Over Finite Chain Rings}



\author{Alexandre Fotue Tabue}
\address{Department of mathematics, Faculty of Sciences,  University of Yaoundé 1, Cameroon} \email{alexfotue@gmail.com}
\author{Christophe Mouaha}
\address{Department of mathematics,  Higher Teachers Training College of Yaoundé, University of Yaoundé 1, Cameroon} \email{cmouaha@yahoo.fr}

\begin{abstract}
Let $\texttt{R}$ be a commutative finite chain ring of invariants
$(q,s),$ and $\Gamma(\texttt{R})$ the Teichmüller's set of
$\texttt{R}.$ In this paper, the trace representation cyclic
$\texttt{R}$-linear codes of length $\ell,$ is presented, when
$\texttt{gcd}(\ell,q)=1.$ We will show that the contractions of
some cyclic $\texttt{R}$-linear codes of length $u\ell$ are
$\gamma$-constacyclic $\texttt{R}$-linear codes of length $\ell,$
where $\gamma\in\Gamma(\texttt{R})\backslash\{0_\texttt{R}\}$ and
the multiplicative order of $\gamma$ is $u.$
\end{abstract}

\maketitle

\section{Introduction}

 Let $\texttt{R}$ be finite chain ring with invariant $(q,s),$
 $\pi:\texttt{R}\rightarrow\mathbb{F}_q$ be the natural
 ring epimorphism, and $\ell$ a positive integer such that $\texttt{gcd}(q,\ell)=1.$ Let
$\texttt{R}^\times$ be the group of units of $\texttt{R},$ and
$\gamma\in\texttt{R}^\times.$ An $\texttt{R}$-linear code
$\mathcal{C}$ of length $\ell$ is \index{constacyclic
code}\emph{$\gamma$-constacyclic} if
$\tau_\gamma(\mathcal{C})=\mathcal{C},$ where
 $\tau_\gamma :  \texttt{R}^\ell \rightarrow
\texttt{R}^{\ell},$ is the $\gamma$-constashift operator, defined
by
$\tau_\gamma(\textbf{c}_0,\textbf{c}_1,\cdots,\textbf{c}_{\ell-1})=(\gamma
\textbf{c}_{\ell-1},\textbf{c}_0,\cdots,\textbf{c}_{\ell-2}).$
Especially, cyclic and negacyclic linear codes correspond to
$\gamma = 1_\texttt{R}$ and $\gamma=-1_\texttt{R},$ respectively
(see \cite{DP04}).  The residue code of $\texttt{R}$-linear code
$\mathcal{C}$ is the $\mathbb{F}_q$-linear code
$\pi(\mathcal{C}):=\left\{(\pi(\textbf{c}_0),\pi(\textbf{c}_1),\cdots,\pi(\textbf{c}_{\ell-1}))\;:\;(\textbf{c}_0,\textbf{c}_1,\cdots,\textbf{c}_{\ell-1})\in\mathcal{C}\right\}.$
The equality
$\pi(\tau_\gamma(\mathcal{C}))=\tau_{\pi(\gamma)}\left(\pi(\mathcal{C})\right),$
enables to see that the residue code of any $\gamma$-constacyclic
$\texttt{R}$-linear code, is an $\pi(\gamma)$-constacyclic
$\mathbb{F}_q$-linear code. In the literature
\cite{Cao13,KZT12,Tap03,Wol99,ZK10}, the class of
$\gamma$-constacyclic $\texttt{R}$-linear codes, which are
studied, have the following property $\gamma\in
1_\texttt{R}+\texttt{R}\theta.$

 In this paper, on the one hand, we will describe each $\gamma$-constacyclic $\texttt{R}$-linear code of length $\ell,$ as
contraction of a cyclic $\texttt{R}$-linear code of length
$u\ell,$ and on the other hand, we will investigate on the
structure of $\gamma$-constacyclic $\texttt{R}$-linear codes,
where $\gamma\in\Gamma(\texttt{R})\setminus\{0_\texttt{R}\}.$

The present paper is organized as follows. In Sect. \ref{Sect:2},
we present results which will be used in the following sections.
Sect. \ref{Sect:3} studies the subring subcode and trace code of a
linear codes over finite chain rings. In Sect.\ref{Sect:4}, the
trace-description of cyclic linear codes over finite chain rings
is presented. For any $\gamma\in\Gamma(\texttt{R}),$ we proceed to
investigate on the structural properties of $\gamma$-constacyclic
codes of arbitrary length $\ell,$ in Sect. \ref{Sect:5}.

\section{Background on finite chain rings}\label{Sect:2}

  Throughout
of this section, $\texttt{R}$ is a commutative ring with identity
and $\texttt{J}(\texttt{R})$ denoted the Jacobson radical of
$\texttt{R},$ and $\texttt{R}^\times$ denotes the multiplicative
group of units of $\texttt{R}.$ The definitions and results on the
finite chain rings are extracted in monographs
\cite{McD74,Nechaev}.

\begin{Definition} We say that $\texttt{R}$ is a \emph{finite chain ring of invariants $(q,s),$}
if:
\begin{enumerate}
    \item $\texttt{R}$ is local principal ideal ring;
    \item $\texttt{R}/\texttt{J}(\texttt{R})\simeq\mathbb{F}_q$ and
$\texttt{R}\supsetneq \texttt{R}\theta \supsetneq \cdots\supsetneq
\texttt{R}\theta^{s-1} \supsetneq \texttt{R}\theta^s =\{0\},$
where $\theta$ is a generator of $\texttt{J}(\texttt{R}).$
\end{enumerate}
\end{Definition}

The map $\pi : \texttt{R}\rightarrow \mathbb{F}_q$ denotes the
canonical projection.

\begin{Lemma}\label{pr-cr} Let $\texttt{R}$ be a finite chain ring of invariants
$(q,s),$ and $\theta$ be a generator of $\texttt{J}(\texttt{R}).$
Then
\begin{enumerate}
    \item $\texttt{R}^\times=\texttt{R}\setminus\texttt{J}(\texttt{R}),$ and the ideals of $\texttt{R}$ are precisely $\texttt{J}(\texttt{R})^t=\texttt{R}\theta^t,$ where $t\in\{0,1,\cdots,s\};$
    \item  $|\texttt{R}^\times|=q^{(s-1)}(q-1)$ and $|\texttt{J}(\texttt{R})^t|=q^{s-t},$ for every $t\in\{0,1,\cdots,s\}.$
\end{enumerate}
\end{Lemma}

\begin{Theorem}\label{thm-cr} Let $\texttt{R}$ be a finite chain ring of invariants
$(q,s),$ and $\theta$ be a generator of $\texttt{J}(\texttt{R}).$
Then
\begin{enumerate}
    \item $\texttt{R}^\times=\Gamma(\texttt{R})^*\cdot(1+\texttt{R}\theta),$ and $\Gamma(\texttt{R})^*\simeq\mathbb{F}_q\setminus\{0\}$ (as multiplicative group) where $\Gamma(\texttt{R})^*:=\{ \textbf{b}\in \texttt{R}\,:\, \textbf{b}\neq 0,\, \textbf{b}^{q}=\textbf{b} \};$
    \item $\Gamma(\texttt{R})^*$ is a cyclic subgroup of $\texttt{R}^\times,$ of order $q-1$ and $|1_\texttt{R}+\texttt{R}\theta|=q^{s-1};$
    \item for every element $\textbf{a}\in\texttt{R},$ there exists a unique $(\textbf{a}_0,\textbf{a}_1,\cdots,\textbf{a}_{s-1})\in\Gamma(\texttt{R})^s,$ such that $\textbf{a}=\textbf{a}_0+\textbf{a}_1\theta+\cdots+\textbf{a}_{s-1}\theta^{s-1}.$
\end{enumerate}
\end{Theorem}

\begin{Definition} Let $\texttt{R}$ be a finite chain ring of invariants
$(q,s),$ and $\theta$ be a generator of $\texttt{J}(\texttt{R}).$
The set $\Gamma(\texttt{R})=\Gamma(\texttt{R})^*\cup\{0\}$ is
called the \emph{Teichmüller set} of $\texttt{R}.$
\end{Definition}

We say that the ring $\texttt{S}$ is an \emph{extension} of
$\texttt{R}$ and we denote it by $\texttt{S}|\texttt{R}$ if
$\texttt{R}$ is a subring of $\texttt{S}$ and $1_\texttt{R} =
1_\texttt{S}.$ We denote by
$\texttt{rank}_\texttt{R}(\texttt{S}),$ the rank of
$\texttt{R}$-module $\texttt{S}.$ We denote by
$\texttt{Aut}_\texttt{R}(\texttt{S}),$ the group of ring
automorphisms of $\texttt{S}$ which fix the elements of
$\texttt{R}.$

\begin{Definition} Let $\texttt{R}$ be a finite chain ring of invariants
$(q,s).$ We say that the finite chain ring $\texttt{S}$ is the
\emph{Galois extension} of $\texttt{R}$ of degree $m,$ if
\begin{enumerate}
    \item $\texttt{S}|\texttt{R}$ is unramified, i.e. $\texttt{J}(\texttt{S})=\texttt{J}(\texttt{R})\texttt{S};$
    \item $\texttt{S}|\texttt{R}$ is normal, i.e. $\texttt{R}:=\{\textbf{a}\in\texttt{S}\;:\;\varrho(\textbf{a})=\textbf{a}\text{ for all }\varrho\in\texttt{Aut}_\texttt{R}(\texttt{S})\}.$
\end{enumerate}
\end{Definition}

\begin{Proposition} Let $\texttt{R}$ be a finite chain ring of invariants $(q,s).$
Let $\texttt{S}$ is the Galois extension of $\texttt{R}$ of degree
$m.$ Then
\begin{enumerate}
    \item $\texttt{S}$ is a free $\texttt{R}$-module of rank $m;$
    \item $\texttt{Aut}_\texttt{R}(\texttt{S})$ is cyclic of order $m;$
    \item $\texttt{S}=\texttt{R}[\xi]$  where $\xi$ is a generator of $\Gamma(\texttt{S}).$
\end{enumerate}
\end{Proposition}

\begin{Definition} Let $\texttt{S}|\texttt{R}$ be the Galois extension of finite
chain rings of degree $m$ and $\sigma$ be a generator of
$\texttt{Aut}_\texttt{R}(\texttt{S}).$ The map
$\texttt{Tr}_\texttt{R}^\texttt{S}:=\sum\limits_{i=0}^{m-1}\sigma^i,$
is called the \emph{trace map} of the Galois extension
$\texttt{S}|\texttt{R}.$
\end{Definition}

\begin{Proposition}\cite[Chap. XIV]{McD74} Let $\texttt{S}|\texttt{T}$ and $\texttt{R}|\texttt{T}$ be Galois extensions of finite
chain rings. Then
\begin{enumerate}
    \item $\texttt{R}=\left\{\textbf{a}\in\texttt{S}\;:\;\sigma(\textbf{a})=\textbf{a}\;\text{ for all } \sigma\in \texttt{Aut}_\texttt{R}(\texttt{S})\right\};$
    \item the bilinear form $\varphi: (\textbf{a},\textbf{b})\mapsto\texttt{Tr}_\texttt{R}^\texttt{S}(\textbf{a}\textbf{b})$ is nondegenerate;
    \item $\texttt{Tr}_\texttt{R}^\texttt{S}$ is a
generator of $\texttt{S}$-module
$\texttt{Hom}_\texttt{R}(\texttt{S},\texttt{R}),$ and
$\texttt{Tr}_\texttt{T}^\texttt{R}\circ\texttt{Tr}_\texttt{R}^\texttt{S}=\texttt{Tr}_\texttt{T}^\texttt{S}.$
\end{enumerate}
\end{Proposition}

\section{Linear codes over finite chain rings}\label{Sect:3}

Recall that an $\texttt{R}$-linear code of length $\ell$ is an
$\texttt{R}$-submodule of $\texttt{R}^\ell.$ We say that an
$\texttt{R}$-linear code is \index{free linear code}\emph{free} if
it is a free as $\texttt{R}$-module.

\subsection{Type and rank of a linear code}

  A matrix $G$ is called a \index{generator matrix}\emph{generator matrix} for $\mathcal{C}$
if the rows of $G$ span $\mathcal{C}$ and none of them can be
written as an $\texttt{R}$-linear combination of the other rows of
$G.$ We say that $G$ is a generator matrix in
\index{standard form}\emph{standard form} if \begin{align}\label{stG}G=\left(%
\begin{array}{cccccc}
  I_{k_{0}}  & G_{0,1}        & G_{0,2}        &\cdots    & G_{0,s-1}                 & G_{0,s}                \\
  0        &  \theta I_{k_{1}} &   \theta G_{1,2} & \cdots   &  \theta G_{1,s-1}          & \theta  G_{1,s}         \\
  \cdots   & \cdots         & \cdots         & \cdots   &  \cdots                   & \cdots               \\
  0        &      0         &           0    & \cdots        & \theta^{s-1} I_{k_{s-1}}  &  \theta^{s-1} G_{s-1,s}
\end{array}%
\right)U,
\end{align}
where $U$ is a suitable permutation matrix. The $s$-tuple $(k_{0},
k_{1},\cdots,k_{s-1})$ is called \index{type of a
matrix}\emph{type} of $G$ and
$\texttt{rank}(G):=k_0+k_1+\cdots+k_{s-1}$ is the \index{rank of a
matrix}\emph{rank} of $G.$

\begin{Proposition}(\cite[Proposition 3.2, Theorem 3.5]{NS00}) Each $\texttt{R}$-linear
code $\mathcal{C}$ admits a generator matrix $G$ standard form.
Moreover, the type is the same for any generator matrix in
standard form for $\mathcal{C}.$
\end{Proposition}

So the type and the rank are the invariants of $\mathcal{C},$ and
henceforth we have the following definition.

\begin{Definition} Let $\mathcal{C}$ be an $\texttt{R}$-linear code.
\begin{enumerate}
    \item The \index{type of a linear code}\emph{type} of $\mathcal{C}$ is the type of a generator matrix of $\mathcal{C}$ in standard form.
    \item The \index{rank of a linear code}\emph{rank} of $\mathcal{C},$ denoted $\texttt{rank}_\texttt{R}(\mathcal{C}),$  is the rank of a generator matrix of $\mathcal{C}$ in standard form.
\end{enumerate}
\end{Definition}

Obviously, any $\texttt{R}$-linear code $\mathcal{C}$ of length
$\ell$ and of type $(k_{0}, k_{1},\cdots,k_{s-1})$ is free if and
only if the rank of $\mathcal{C}$ is $k_0,$ and
$k_1=k_2=\cdots=k_{s-1}=0.$ It defines the scalar product on
$\texttt{R}^\ell$ by:
$\underline{\textbf{a}}\cdot\underline{\textbf{b}}^{\texttt{T}}:=\sum\limits_{i=0}^{\ell-1}\textbf{a}_i\textbf{b}_i,$
where $\underline{\textbf{b}}^{\texttt{T}}$ is the transpose of
$\underline{\textbf{b}}.$ Let $\mathcal{C}$ be an
$\texttt{R}$-linear code of length $\ell.$ The dual code of
$\mathcal{C},$ denoted $\mathcal{C}^\perp,$ is an
$\texttt{R}$-linear code of length $\ell,$ define by:
$\mathcal{C}^\perp:=\left\{\underline{\textbf{a}}\in\texttt{R}^\ell\;:\;\underline{\textbf{a}}\cdot\underline{\textbf{b}}^{\texttt{T}}=0\text{
for all }\underline{\textbf{c}}\in\mathcal{C}\right\}.$ A
generator matrix of $\mathcal{C}^\perp,$ is called parity-check
matrix of $\mathcal{C}.$

\begin{Proposition}\label{dual-type}(\cite[Theorem 3.10]{NS00})
Let $\mathcal{C}$ be an $\texttt{R}$-linear code of length $\ell$
and of type $(k_{0}, k_{1},\cdots,k_{s-1}).$  Then
\begin{enumerate}
    \item the type of $\mathcal{C}^\perp$ is $(\ell-k,
k_{s-1},\cdots,k_{1}),$ where $k:=k_{0}+k_{1}+\cdots+k_{s-1}.$
    \item $|\mathcal{C}|=q^{\sum\limits_{t=0}^{s-1}(s-t)k_t},$
    where $|\mathcal{C}|$ denotes the number of elements of $\mathcal{C}.$
\end{enumerate}
\end{Proposition}

\subsection{Galois closure of a linear code over a finite chain ring}

  Let $\mathcal{B}$ be an
$\texttt{S}$-linear codes of length $\ell.$ Then
$$\sigma(\mathcal{B}):=\left\{(\sigma(\textbf{c}_0),\cdots,\sigma(\textbf{c}_{\ell-1}))\;:\;(\textbf{c}_0,\cdots,\textbf{c}_{\ell-1})\in\mathcal{B}\right\}$$
is also an $\texttt{S}$-linear codes of length $\ell.$  We say
that the $\texttt{S}$-linear code $\mathcal{B}$ is called
\index{Galois invariant}\emph{$\sigma$-invariant} if
$\sigma(\mathcal{B}) = \mathcal{B}.$ The \index{subring
subcode}\emph{subring subcode} of $\mathcal{B}$ to $\texttt{R},$
is $\texttt{R}$-linear code
$\texttt{Res}_\texttt{R}(\mathcal{B}):=\mathcal{B}\cap
\texttt{R}^{\ell},$ and the \index{trace code}\emph{trace code} of
$\mathcal{B}$ over $\texttt{R},$ is the $\texttt{R}$-linear code
$$
\texttt{Tr}_\texttt{R}^\texttt{S}(\mathcal{B}):=\left\{(\texttt{Tr}_\texttt{R}^\texttt{S}(\textbf{c}_0),\cdots,\texttt{Tr}_\texttt{R}^\texttt{S}(\textbf{c}_{\ell-1}))\;:\;(\textbf{c}_0,\cdots,\textbf{c}_{\ell-1})\in\mathcal{B}\right\}.
$$
It is clear that
$\texttt{Tr}_\texttt{R}^\texttt{S}(\sigma(\mathcal{B}))=\texttt{Tr}_\texttt{R}^\texttt{S}(\mathcal{B}).$
The \emph{extension code} of an $\texttt{R}$-linear code
$\mathcal{C}$ to $\texttt{S},$ is the $\texttt{S}$-linear code
$\texttt{Ext}_\texttt{S}(\mathcal{C}),$  formed by taking all
combinations of codewords of $\mathcal{C}.$ The following theorem
generalizes Delsarte's celebrated result (see \cite[Ch.7.\S8.
Theorem 11.]{WS77}).

\begin{Theorem} (\cite[Theorem 3]{MNR13}). Let $\mathcal{B}$ be an $\texttt{S}$-linear code then
$\texttt{Tr}_\texttt{R}^\texttt{S}(\mathcal{B}^\perp)
=\texttt{Res}_\texttt{R}(\mathcal{B})^\perp,$ where
$\mathcal{B}^\perp$ is the dual to $\mathcal{B}$ with respect to
the usual scalar product, and
$\texttt{Res}_\texttt{R}(\mathcal{B})^\perp$ is the dual of
$\texttt{Res}_\texttt{R}(\mathcal{B})$ in $\texttt{R}^\ell.$
\end{Theorem}

\begin{Definition}\label{G-inv} Let $\mathcal{B}$ be an $\texttt{S}$-linear
code. The \index{Galois closure}\emph{$\sigma$-closure} of
$\mathcal{B},$ is the smallest $\sigma$-invariant
$\texttt{S}$-linear code $\widetilde{\mathcal{B}},$ containing
$\mathcal{B}.$
\end{Definition}

\begin{Proposition}\label{tra-inv} Let $\mathcal{B}$ be an $\texttt{S}$-linear code.
Then
$\widetilde{\mathcal{B}}=\sum\limits_{i=0}^{m-1}\sigma^i(\mathcal{B})$
and
$\texttt{Tr}_{\texttt{R}}^{\texttt{S}}(\mathcal{B})=\texttt{Tr}_{\texttt{R}}^{\texttt{S}}(\widetilde{\mathcal{B}}).$
\end{Proposition}

\begin{Proof}  We have
 $\mathcal{B}\subseteq\widetilde{\mathcal{B}}$ and $\sigma(\widetilde{\mathcal{B}})=\widetilde{\mathcal{B}},$ by Definition\,\ref{G-inv} of
$\widetilde{\mathcal{B}}.$ So
$\sigma^i(\mathcal{B})\subseteq\widetilde{\mathcal{B}},$ for all
$i\in\{0,1,\cdots,m-1\}.$ Hence
$\sum\limits_{i=0}^{m-1}\sigma^i(\mathcal{B})\subseteq\widetilde{\mathcal{B}}.$
Since
$\sigma\left(\sum\limits_{i=0}^{m-1}\sigma^i(\mathcal{B})\right)=\sum\limits_{i=0}^{m-1}\sigma^i(\mathcal{B})$
and
$\mathcal{B}\subseteq\sum\limits_{i=0}^{m-1}\sigma^i(\mathcal{B}),$
as $\widetilde{\mathcal{B}}$ is the smallest $\texttt{S}$-linear
code containing $\mathcal{B},$ which is $\sigma$-invariant, it
follows
$\widetilde{\mathcal{B}}\subseteq\sum\limits_{i=0}^{m-1}\sigma^i(\mathcal{B}).$
Hence
$\widetilde{\mathcal{B}}=\sum\limits_{i=0}^{m-1}\sigma^i(\mathcal{B}).$
Thanks to \cite[Proposition 1.]{MNR13},
 $\texttt{Tr}_{\texttt{R}}^{\texttt{S}}(\widetilde{\mathcal{B}}) =\texttt{Tr}_{\texttt{R}}^{\texttt{S}}(\mathcal{B}).$
\end{Proof}

The following Theorem summarizes the obtained results  in
\cite{MNR13}.

\begin{Theorem}\label{trace} Let $\mathcal{B}$ be an $\texttt{S}$-linear code and $\sigma$ be a generator of $\texttt{Aut}_{\texttt{R}}(\texttt{S})$. Then the following statements are equivalent:
\begin{enumerate}
     \item $\mathcal{B}$ is  $\sigma$-invariant;
    \item $\texttt{Tr}_{\texttt{R}}^{\texttt{S}}(\mathcal{B})=\texttt{Res}_\texttt{R}(\mathcal{B});$
    \item $\mathcal{B},$ and $\texttt{Res}_\texttt{R}(\mathcal{B})$  have the same type.
\end{enumerate}
\end{Theorem}

\begin{Proof} Let $\mathcal{B}$ be an $\texttt{S}$-linear code.
\begin{description}
    \item[$1.\Leftrightarrow 2.$] Thanks to \cite[Theorem 2]{MNR13}.
    \item[$1.\Leftrightarrow 3.$] Since any $\texttt{R}$-basis of $\texttt{Res}_\texttt{R}(\mathcal{B})$ is also an
$\texttt{S}$-basis of
$\texttt{Ext}_\texttt{S}\left(\texttt{Res}_\texttt{R}(\mathcal{B})\right).$
Thanks to \cite[Theorem 1]{MNR13}, we deduce that
$\mathcal{B}=\texttt{Ext}_\texttt{S}\left(\texttt{Tr}_{\texttt{R}}^{\texttt{S}}(\mathcal{B})\right)$
if and only if $\mathcal{B}$ and
$\texttt{Res}_\texttt{R}(\mathcal{B})$ have the same type.
\end{description}
\end{Proof}

\section{Cyclic linear codes over finite chain rings}\label{Sect:4}

Let $\ell$ be a positive integer such that
$\texttt{gcd}(q,\ell)=1.$ Then the remainder
$q\,(\texttt{mod}\;\ell)$ of $q$ by $\ell,$ belongs to
$\mathbb{Z}_\ell^\times,$ the positive integer $m$ denotes the
multiplicative order of $q\,(\texttt{mod}\;\ell).$ Let
$\Sigma_\ell:=\{0,1,\cdots,\ell-1\}$ be the underling set of
$\mathbb{Z}_\ell.$

\subsection{Cyclotomic cosets}

Let $u$ be a positive integer. The set of \emph{multiples} of $u$
in $\mathrm{A}$ is
$$u\mathrm{A}:=\{uz\,(\texttt{mod}\,\ell)\;:\;z\in \mathrm{A}\}.$$
The \index{Galois closure of a set} $q$-\emph{closure} of
$\mathrm{A}$ is $\complement_q(\mathrm{A}):=\underset{i\in
\mathbb{N}}{\cup}q^i\mathrm{A}.$

\begin{Definition} Let $z\in\Sigma_\ell.$ The \index{cyclotomic
coset}\emph{$q$-cyclotomic coset modulo $\ell,$} containing $z,$
the Galois closure of $\{z\}.$ We simply write
$\complement_q(z):=\complement_q(\{z\}).$
\end{Definition}

It denotes by $\Re_\ell(q)$ the set of $q$-closure subsets of
$\Sigma_\ell.$ Obviously, the $q$-cyclotomic cosets modulo $\ell,$
form a partition of $\Sigma_\ell.$ Let $\Sigma_\ell(q)$ be a set
of representatives of each $q$-cyclotomic cosets modulo $\ell.$

\begin{Proposition}\cite[Proposition 5.2]{BGG14}\label{ct-q} We have  $|\Sigma_\ell(q)|
=\sum\limits_{d|\ell}\frac{\phi(\ell)}{\texttt{ord}_\ell(q)},$
where $\phi(.)$ is the Euler totient function and
$\texttt{ord}_\ell(q):=\texttt{min}\left\{i\in\mathbb{N}\;:\;q^{i+1}\,\equiv\,1\,(\texttt{mod}\;\ell)\right\}.$
\end{Proposition}

\begin{Notation} Let $z\in\Sigma_\ell$ and $\mathrm{A}$ be a
subset of $\Sigma_\ell$ and $u\in\mathbb{N},$ with
$\texttt{gcd}(u\,,\,q)=1.$
\begin{enumerate}
     \item The \index{opposite of an element}\emph{opposite} of $\mathrm{A}$ is $-\mathrm{A}:=\{\ell-z\;:\;z\in \mathrm{A}\}.$
     \item The \index{complementary of a set }\emph{complementary} of $\mathrm{A}$ is $\overline{\mathrm{A}}:=\left\{z\in\Sigma_\ell\;:\;z\not\in \mathrm{A}\right\}.$
     \item The \index{dual of a set}\emph{dual} of $\mathrm{A}$ is $\mathrm{A}^{\diamond}:=\overline{-\mathrm{A}}.$
\end{enumerate}
\end{Notation}

\begin{Remark} Let $\mathrm{A}$ be a
subset of $\Sigma_\ell.$ Then
$\complement_q\left(\overline{\mathrm{A}}\right)=\overline{\complement_q(\mathrm{A})}$
and $-\complement_q(\mathrm{A})=\complement_q(-\mathrm{A}).$
Moreover $(\mathrm{A}^{\diamond})^{\diamond}=\mathrm{A}.$
\end{Remark}

\begin{Example} We take $\;\ell=20, q=3.$ The $q$-cyclotomic cosets modulo
$\ell,$ are: $ \complement_q(\{0\})=\{0\},
\complement_q(\{5\})=\{5,15\}, \complement_q(\{10\})=\{10\},$ and
$$\begin{array}{ll}
 \complement_q(\{1\})=\{1,3,9,7\} ; & \complement_q(\{2\})=\{2,6,18,14\}; \\
 \complement_q(\{4\})=\{4,12,16,8\}; &  \complement_q(\{11\})=\{11,13,19,17\}.
\end{array}$$
 So $\Sigma_\ell(q)=\{0,1,2,4,5,10,11\}.$ We remark
that $\complement_q(\{-z\})=\complement_q(\{z\}),$ for every
$z\in\{0,2,4,5,10\}.$ We set $\;\mathrm{I}:=[0,10].$ We have
$\mathrm{A}:=\complement_q(\mathrm{I})=\complement_q(\{0,1,2,4,5,10\}),$
$-\mathrm{A}=\complement_q(\{2,4,5,10,11\}),$ and
$\mathrm{A}^\diamond:=\complement_q(\{1\}).$
\end{Example}

\subsection{Likewise Reed-Solomon codes over finite chain rings}

Let $\texttt{S}$ be the Galois extension of $\texttt{R}$ of degree
$m$ and $\xi$ be a generator of
$\Gamma(\texttt{S})\backslash\{0\}.$  Let
$\mathrm{A}:=\{a_1,a_2,\cdots,a_k\}$ be a subset of $\Sigma_\ell.$
One denotes by $\textbf{P}(\texttt{S}\,;\,\mathrm{A}),$ the free
$\texttt{S}$-module with $\texttt{S}$-basis
$\{X^a\,:\,a\in\mathrm{A}\}.$ Since $m$ is the smallest positive
integer with $q^m\equiv 1   \,(\texttt{mod}\,\ell),$ we can write
$\eta:=\xi^{\frac{q^m-1}{\ell}}$ and the multiplicative order of
$\eta$ is  $\ell.$ The evaluation
$$\begin{array}{cccc}
  \textbf{ev}_\eta: & \textbf{P}(\texttt{S}\,;\,\mathrm{A}) & \rightarrow & \texttt{S}^\ell \\
   & f & \mapsto &
   (f(1),f(\eta),\cdots,f(\eta^{\ell-1})),
\end{array}$$ is an $\texttt{S}$-modules monomorphism. We see that if $\mathrm{A}:=\{0,1,\cdots,k-1\},$ then for any  $\ell^{\texttt{th}}$-primitive root of unity
$\eta$ in $\Gamma(\texttt{S}),$ the $\texttt{S}$-linear code
$\textbf{ev}_\eta(\textbf{P}(\texttt{S}\,;\,\mathrm{A}))$ is a
primitive Reed-Solomon code. For this reason, we define
 Likewise Reed-Solomon codes which are a family of codes defined over large finite chain rings as follows.

\begin{Definition} Let  $\mathrm{A}$ be a subset of $\Sigma_\ell,$ and $\texttt{S}$ be a finite chain ring such that $|\Gamma(\texttt{S})|\geq\ell.$
Let $\eta\in\Gamma(\texttt{S})$ and the multiplicative order of
$\eta$ is $\ell.$ The $\texttt{S}$-submodule
$\textbf{ev}_\eta(\textbf{P}(\texttt{S}\,;\,\mathrm{A}))$ is
called \index{likewise Reed-Solomon code}\emph{likewise
Reed-Solomon code} over $\texttt{S},$ with defining pair
$(\eta\;,\;\mathrm{A}).$
\end{Definition}

We remark that
$\textbf{L}_\eta(\texttt{S}\,;\,\mathrm{A}):=\textbf{ev}_\eta(\textbf{P}(\texttt{S}\,;\,\mathrm{A}))$
is the free $\texttt{S}$-linear code with free $\texttt{S}$-basis
$\{\textbf{ev}_\eta(X^a)\;:\;a\in\mathrm{A}\},$ where $\mathrm{A}$
is a subset of $\Sigma_\ell.$ We remark that
$\textbf{L}_\eta(\texttt{S}\,;\,\emptyset)=\{\textbf{0}\},$
$\textbf{L}_\eta(\texttt{S}\,;\,\{0\})=\textbf{1}$ and
$\textbf{L}_\eta(\texttt{S}\,;\,\Sigma_\ell)=\texttt{S}^\ell.$

\begin{Proposition}\label{cyclic} Let $\mathrm{A},\mathrm{B}$ be two subsets of
$\Sigma_\ell.$  Then
\begin{enumerate}
    \item $\textbf{L}_\eta(\texttt{S}\,;\,\mathrm{A})$ is cyclic;
    \item $\textbf{L}_{\eta}(\texttt{S}\,;\,\mathrm{A}\cup\mathrm{B})=\textbf{L}_\eta(\texttt{S}\,;\,\mathrm{A})+\textbf{L}_\eta(\texttt{S}\,;\,\mathrm{B})$ and $\textbf{L}_{\eta}(\texttt{S}\,;\,\mathrm{A}\cap\mathrm{B})=\textbf{L}_\eta(\texttt{S}\,;\,\mathrm{A})\cap\textbf{L}_\eta(\texttt{S}\,;\,\mathrm{B}).$
\end{enumerate}
\end{Proposition}

\begin{proof} Consider the codeword $\underline{\textbf{c}}_a =
\left(1,\eta^a, \cdots,\eta^{a(\ell-1)}\right).$ Then the shift of
$\underline{\textbf{c}}_a$ is $\eta^{-a}\underline{\textbf{c}}_a.$
Since $\textbf{L}_\eta(\texttt{S}\,;\,\mathrm{A})$ is
$\texttt{S}$-linear, we have
$\eta^{-a}\underline{\textbf{c}}_a\in\textbf{L}_\eta(\texttt{S}\,;\,\mathrm{A}).$
Hence $\textbf{L}_\eta(\texttt{S}\,;\,\mathrm{A})$ is cyclic. It
is clear that
$\textbf{L}_{\eta}(\texttt{S}\,;\,\mathrm{A}\cup\mathrm{B})\supseteq\textbf{L}_\eta(\texttt{S}\,;\,\mathrm{A})+\textbf{L}_\eta(\texttt{S}\,;\,\mathrm{B}).$
The set
$\{\textbf{ev}_\eta(X^a)\;:\;a\in\mathrm{A}\cup(\mathrm{B}\setminus\mathrm{A})\}$
is a free $\texttt{R}$-basis of
$\textbf{L}_{\eta}(\texttt{S}\,;\,\mathrm{A}\cup\mathrm{B})$ and
$\textbf{L}_\eta(\texttt{S}\,;\,\mathrm{A})+\textbf{L}_\eta(\texttt{S}\,;\,\mathrm{B}).$
Hence,
$\textbf{L}_{\eta}(\texttt{S}\,;\,\mathrm{A}\cup\mathrm{B})=\textbf{L}_\eta(\texttt{S}\,;\,\mathrm{A})+\textbf{L}_\eta(\texttt{S}\,;\,\mathrm{B}).$
We leave the last equality as an exercise.
\end{proof}

\begin{Proposition}\label{L-opera} Let $\mathrm{A}$ be a subset of
$\Sigma_\ell$ and $u$ be a positive integer such that
$\texttt{gcd}(\ell,u)=1.$ Then
\begin{enumerate}
    \item $\textbf{L}_{\eta^u}(\texttt{S}\,;\,\mathrm{A})=\textbf{L}_\eta(\texttt{S}\,;\,u\mathrm{A});$
    \item $\textbf{L}_\eta(\texttt{S}\,;\,\mathrm{A})^\perp=\textbf{L}_\eta(\texttt{S}\,;\,\mathrm{A}^{\diamond});$
    \item $\textbf{L}_\eta\left(\texttt{S}\,;\,\complement_q(\mathrm{A})\right)$ is the $\sigma$-closure of $\textbf{L}_\eta(\texttt{S}\,;\,\mathrm{A}).$
\end{enumerate}
\end{Proposition}

\begin{Proof} Assume that $\texttt{gcd}(\ell,u)=1.$ Then $\eta$ and $\eta^u$
are $\ell^{\texttt{th}}$-primitive roots of unity. Since
$\{\textbf{ev}_\eta(X^a)\;:\;a\in u\mathrm{A}\}$ is a free
$\texttt{R}$-basis of
$\textbf{L}_{\eta^u}(\texttt{S}\,;\,\mathrm{A}),$ we have
$\textbf{L}_{\eta^u}(\texttt{S}\,;\,\mathrm{A})=\textbf{L}_\eta(\texttt{S}\,;\,u\mathrm{A}).$

  A free $\texttt{S}$-basis of
$\textbf{L}_\eta(\texttt{S}\,;\,\mathrm{A}^{\diamond})$ is
$\{\underline{\textbf{c}}_a\;:\;-a\in \overline{\mathrm{A}}\}$
where
$\underline{\textbf{c}}_a:=(1,\eta^{-a},\cdots,\eta^{-a(\ell-1)})\in\textbf{L}_\eta(\texttt{S}\,;\,\mathrm{A}^{\diamond}).$
Then for all $b\in \mathrm{A},$
$\underline{\textbf{c}}_b:=(1,\eta^{b},\cdots,\eta^{b(\ell-1)})\in\textbf{L}_\eta(\texttt{S}\,;\,\mathrm{A}),$
we have
$\underline{\textbf{c}}_b\underline{\textbf{c}}_a^{\texttt{tr}}=\sum\limits_{j=0}^{\ell-1}\eta^{(b-a)j}.$
It is easy to check that $\sum\limits_{j=0}^{\ell-1}\eta^{ij}=0,$
when $i\not\equiv\, 0(\texttt{mod}\,\ell).$ Since
 $0<b-a<\ell,$ we have $\underline{\textbf{c}}_b\underline{\textbf{c}}_a^{\texttt{tr}}=0.$ So
$\textbf{L}_\eta(\texttt{S}\,;\,\mathrm{A}^{\diamond})\subseteq\textbf{L}_\eta(\texttt{S}\,;\,\mathrm{A})^\perp.$
Comparison of cardinality yields
$\textbf{L}_\eta(\texttt{S}\,;\,\mathrm{A})^\perp=\textbf{L}_\eta(\texttt{S}\,;\,\mathrm{A}^{\diamond}).$
Finally,
$\sigma(\textbf{L}_\eta(\texttt{S}\,;\,\mathrm{A}))=\textbf{L}_\eta(\texttt{S}\,;\,q\mathrm{A}).$
    So by Proposition\,\ref{tra-inv}, we have
$$\widetilde{\textbf{L}_\eta(\texttt{S}\,;\,\mathrm{A})}=\sum\limits_{i=0}^{m-1}\textbf{L}_\eta(\texttt{S}\,;\,q^i\mathrm{A})=\textbf{L}_\eta\left(\texttt{S}\,;\,\overset{m-1}{\underset{i=0}{\bigcup}}q^i\mathrm{A}\right).$$
    Since $\complement_q(\mathrm{A})=\overset{m-1}{\underset{i=0}{\bigcup}}q^i\mathrm{A},$ we obtain
    $\widetilde{\textbf{L}_\eta(\texttt{S}\,;\,\mathrm{A})}=\textbf{L}_\eta(\texttt{S}\,;\,\complement_q(\mathrm{A})).$
\end{Proof}

\subsection{Trace representation of free cyclic linear codes}

We introduce the map trace-evaluation
 $\texttt{Tr}_\texttt{R}^\texttt{S}\circ\textbf{ev}_\eta:\textbf{P}_\eta(\texttt{S};\mathrm{A})\rightarrow \texttt{R}^\ell,$ defined
 by:
 $$\texttt{Tr}_\texttt{R}^\texttt{S}\circ\textbf{ev}_\eta(X^a):=\texttt{Tr}_\texttt{R}^\texttt{S}\left(1,\eta^a,\cdots,\eta^{a(\ell-1)}\right),$$
 for all $a\in \mathrm{\mathrm{A}}.$ In the sequel, we write:
 $\mathrm{\textbf{C}}_\eta(\texttt{R}\,;\,\mathrm{A}):=\texttt{Tr}_\texttt{R}^\texttt{S}\left(\textbf{L}_\eta(\texttt{S}\,;\,\mathrm{A})\right),$
  and  $\mathrm{\textbf{C}}_\eta(\texttt{R}\,;\,\mathrm{A})$ is a free cyclic $\texttt{R}$-linear code of length $\ell.$
 The immediate proprieties of trace representation of free cyclic
 linear codes over finite chain ring are given in the following.

 \begin{Proposition}\label{c(a)}
 Let $\mathrm{\mathrm{A}}, \mathrm{B}$ be two empty subsets of $\Sigma_\ell.$  Then
 \begin{enumerate}
    \item  $\mathrm{\textbf{C}}_\eta(\texttt{R}\,;\,\mathrm{A})=\mathrm{\textbf{C}}_\eta\left(\texttt{R}\,;\,\complement_q(\mathrm{A})\right);$
     \item $\texttt{rank}_{\texttt{S}}(\textbf{L}_\eta(\texttt{S}\,;\,\complement_q(\mathrm{A})))=|\complement_q(\mathrm{A})|$ and $\mathrm{\textbf{C}}_\eta(\texttt{R}\,;\,\mathrm{A})^{\perp}=\mathrm{\textbf{C}}_\eta(\texttt{R}\,;\,\mathrm{A}^{\diamond});$
     \item $\textbf{C}_{\eta}(\texttt{S}\,;\,\mathrm{A}\cup\mathrm{B})=\textbf{C}_\eta(\texttt{S}\,;\,\mathrm{A})+\textbf{C}_\eta(\texttt{S}\,;\,\mathrm{B})$ and $\textbf{C}_{\eta}(\texttt{S}\,;\,\mathrm{A}\cap\mathrm{B})=\textbf{C}_\eta(\texttt{S}\,;\,\mathrm{A})\cap\textbf{C}_\eta(\texttt{S}\,;\,\mathrm{B}).$
      \end{enumerate}
  \end{Proposition}

\begin{Proof} Let $\mathrm{A}, \mathrm{B }$ be two subsets of
$\Sigma_\ell.$
\begin{enumerate}
    \item From Proposition\,\ref{tra-inv},
$\mathrm{\textbf{C}}_\eta(\texttt{R}\,;\,\mathrm{A})=\texttt{Tr}(\textbf{L}_\eta(\texttt{S}\,;\,\mathrm{A}))=\texttt{Tr}(\textbf{L}_\eta(\texttt{S}\,;\,\complement_q(\mathrm{A})))=\mathrm{\textbf{C}}_\eta(\texttt{R}\,;\,\complement_q(\mathrm{A})).$
    \item Theorem\,\ref{trace}(3) yields $\mathrm{\textbf{C}}_\eta(\texttt{R}\,;\,\mathrm{A})=\texttt{Tr}(\textbf{L}_\eta(\texttt{S}\,;\,\complement_q(\mathrm{A})))=\texttt{Res}_\texttt{R}(\textbf{L}_\eta(\texttt{S}\,;\,\complement_q(\mathrm{A}))).$
    So $$\texttt{rank}_{\texttt{R}}(\mathrm{\textbf{C}}_\eta(\texttt{R}\,;\,\mathrm{A}))=\texttt{rank}_{\texttt{S}}(\textbf{L}_\eta(\texttt{S}\,;\,\complement_q(\mathrm{A})))=|\complement_q(\mathrm{A})|.$$ From Proposition\,\ref{L-opera},
     $\mathrm{\textbf{C}}_\eta(\texttt{R}\,;\,\mathrm{A})^{\perp}= \mathrm{\textbf{C}}_\eta\left(\texttt{R}\,;\,
     \mathrm{A}^\diamond\right).$
\end{enumerate}
\end{Proof}

The following theorem gives the number of cyclic codes and free
cyclic codes over finite chain rings.

\begin{Lemma}\cite[Theorem 5.1]{BGG14}\label{ct-cy}
Let $\texttt{R}$ be a finite chain ring of invariants $(q,s).$
Then the following holds:
\begin{enumerate}
\item the number of cyclic $\texttt{R}$-linear codes of length
$\ell,$ is equal to $(s+1)^{|\Sigma_\ell(q)|},$
\item the number
of free cyclic $\texttt{R}$-linear codes of length $\ell,$  is
equal to $2^{|\Sigma_\ell(q)|}.$
\end{enumerate}
\end{Lemma}

 \begin{Lemma}\label{irr1} Let $\texttt{R}$ be a finite chain ring of invariants $(q,s)$ and $\texttt{S}$ be the Galois extension of $\texttt{R}$
 of degree $m.$ Let $z\in\Sigma_\ell.$
 Set $\texttt{S}=\texttt{R}[\xi],$ $m_z:=|\complement_q(z)|,$ $\eta:=\xi^{\frac{q^m-1}{\ell}}.$ and  $\zeta:=\eta^{-z}.$  Then the map
$$\begin{array}{cccc}
  \psi_{z}: & \texttt{R}[\xi^{m_z}] & \longrightarrow & \mathrm{\textbf{C}}_\eta\left(\texttt{R}\,;\,\{z\}\right) \\
      & \textbf{a} & \longmapsto & \texttt{Tr}_\texttt{R}^\texttt{S}\left(\emph{ev}_\eta(\textbf{a}X^{z})\right) \\
  \end{array}
$$ is an $\texttt{R}$-module isomorphism. Further $\texttt{R}[\xi^{m_z}]$ is the Galois extension of $\texttt{R}$ of degree $m_z$
and $\psi_{z}\circ t_{\zeta}=\tau_1\circ\psi_{z},$ where
$t_\zeta(\textbf{a})=\textbf{a}\zeta,$ for all
$\textbf{a}\in\texttt{R}[\eta].$
\end{Lemma}

\begin{proof} It is
clear that  $a\in\texttt{Ker}(\psi_{z})$ if and only if
$a\in\texttt{R}[\xi^{m_z}]^{\perp_{\texttt{Tr}}}\cap\texttt{R}[\xi^{m_z}],$
where duality $\perp_{\texttt{Tr}}$ is with respect to trace form.
As the trace bilinear form is nondegenerate, we have
$\texttt{S}=\texttt{R}[\xi^{m_z}]^{\perp_{\texttt{Tr}}}\oplus\texttt{R}[\xi^{m_z}]$
and $\texttt{Ker}(\psi_{z})=\{0\}.$ Hence $\psi_{z}$ is an
$\texttt{R}$-module monomorphism. We remark that,
$\mathrm{\textbf{C}}_\eta(\texttt{R}\,;\,\{z\})$ is cyclic, if and
only if $\psi_{z}\circ t_{\zeta}=\tau_1\circ\psi_{z},$ for all
$a\in\texttt{R}[\eta].$ Finally, we have
$\texttt{S}=\texttt{R}[\xi],$ so $\texttt{R}[\xi^{m_z}]$ is the
Galois extension of $\texttt{R}$ of degree $m_z.$ Hence,
$\psi_{z}$ is an $\texttt{R}$-module isomorphism.
\end{proof}

\begin{Definition} A non trivial cyclic $\texttt{R}$-linear code
$\mathcal{C}$ is said to be \index{irreducible cyclic linear
code}\emph{irreducible}, if for all $\texttt{R}$-linear cyclic
subcodes $\mathcal{C}_1$ and $\mathcal{C}_2$ of $\mathcal{C},$
such that, $\mathcal{C}=\mathcal{C}_1\oplus\mathcal{C}_2,$ implies
$\mathcal{C}_1=\{\textbf{0}\}$ or $ \mathcal{C}_2=\{\textbf{0}\}.$
\end{Definition}

\begin{Proposition}\label{irr(a)}
The irreducible cyclic $\texttt{R}$-linear codes are precisely
$\theta^t\mathrm{\textbf{C}}_\eta(\texttt{R}\,;\,\{z\})$s, where
$t\in\{0,1,\cdots,s-1\}$ and  $z\in\Sigma_\ell(q).$
\end{Proposition}

\begin{Proof} By Lemma \ref{irr1}, the cyclic $\texttt{R}$-linear code $\mathrm{\textbf{C}}_\eta(\texttt{R}\,;\,\{z\}))$  and all the
$\texttt{R}$-linear cyclic subcodes are irreducible.  Let
$\mathcal{C}$ be an irreducible cyclic $\texttt{R}$-linear code.
Then the $\texttt{R}$-linear code
$\texttt{Quot}_{s-1}(\mathcal{C}):=\left\{\textbf{c}\in\texttt{R}^\ell\;:\;\theta^{s-1}\textbf{c}\in\mathcal{C}\right\}$
is cyclic and free, and so
$\texttt{Quot}_{s-1}(\mathcal{C})=\mathrm{\textbf{C}}_\eta(\texttt{R}\,;\,\mathrm{A})$
for some $\mathrm{A}\subset\Sigma_\ell(q)$ and
$\mathrm{A}\neq\emptyset.$ Assume that $|\mathrm{A}|>1.$ Then
$\mathrm{\textbf{C}}_\eta(\texttt{R}\,;\,\mathrm{A})=\mathrm{\textbf{C}}_\eta(\texttt{R}\,;\,\mathrm{A}_1)\oplus\mathrm{\textbf{C}}_\eta(\texttt{R}\,;\,\mathrm{A}_2)$
where $\mathrm{A}_1\cap \mathrm{A}_2=\emptyset,$
$\mathrm{A}_1\neq\emptyset$ and $\mathrm{A}_2\neq\emptyset.$ We
have
$\mathcal{C}\cap\mathrm{\textbf{C}}_\eta(\texttt{R}\,;\,\mathrm{A}_1)\neq\{\textbf{0}\}$
and
$\mathcal{C}\cap\mathrm{\textbf{C}}_\eta(\texttt{R}\,;\,\mathrm{A}_2)\neq\{\textbf{0}\}.$
Therefore
$\mathcal{C}=(\mathcal{C}\cap\mathrm{\textbf{C}}_\eta(\texttt{R}\,;\,\mathrm{A}_1))\oplus(\mathcal{C}\cap\mathrm{\textbf{C}}_\eta(\texttt{R}\,;\,\mathrm{A}_2)).$
It is impossible, because $\mathcal{C}$ be an irreducible. So
$|\mathrm{A}|=1.$ Now,
$\mathcal{C}\subseteq\mathrm{\textbf{C}}_\eta(\texttt{R}\,;\,\{z\}),$
it follows that
$\mathcal{C}=\theta^t\mathrm{\textbf{C}}_\eta(\texttt{R}\,;\,\{z\}),$
for some $t\in\{0,1,\cdots,s-1\}.$
\end{Proof}

We set $\Sigma_\ell(q)$ a set of representatives of each
$q$-cyclotomic cosets modulo $\ell.$ An $(q,s)$-\index{cyclotomic
partition}\emph{cyclotomic partition} modulo $\ell,$ is the
$(s+1)$-tuple $(\mathrm{A}_0,\mathrm{A}_1,\cdots,\mathrm{A}_s)$
with the property
$\mathrm{A}_t=\complement_q\left(\lambda^{-1}(\{t\})\right),$
where $\lambda:\Sigma_\ell(q)\rightarrow\{0,1,\cdots, s\}$ is a
map. Denoted by
$$\Re_\ell(q,s):=\left\{(\mathrm{A}_0,\mathrm{A}_1,\cdots,\mathrm{A}_s) \;\;:\; \left(\exists
\lambda\in\{0,1,\cdots,s\}^{\Sigma_\ell(q)}\right)\left(\mathrm{A}_t=\lambda^{-1}(\{t\})\right)\right\}$$
the set of $(q,s)$-cyclotomic partitions modulo $\ell,$ and
$\texttt{Cy}(\texttt{R},\ell)$ the set of cyclic
$\texttt{R}$-linear codes of length $\ell.$ We have
$|\Re_\ell(q,s)|=(s+1)^{|\Sigma_\ell(q)|}.$

\begin{Example} We take $\;\ell=20, q=3$ and $s=2.$ Then
$|\Sigma_\ell(q)|=13$ and $|\Re_\ell(q,s)|=3^{7}.$ An
$(q,s)$-cyclotomic partition modulo $\ell,$ is
$\underline{\mathrm{A}}:=\left(\complement_q(\{0,1,2\}),\complement_q(\{5,11\}),\complement_q(\{4,10\})\right).$
\end{Example}

\begin{Theorem}\label{decy} Any cyclic $\texttt{R}$-linear code $\mathcal{C}$ there exists a unique $\mathrm{A}:=(\mathrm{A}_0,\mathrm{A}_1,\cdots,\mathrm{A}_s)\in\Re_\ell(q,s)$
such that $\mathcal{C}=\mathrm{\textbf{C}}_\texttt{R}(\mathrm{A})$
and
$\mathrm{\textbf{C}}_\texttt{R}(\mathrm{A})=\overset{s-1}{\underset{t=0}{\bigoplus}}\theta^{t}\mathrm{\textbf{C}}_\eta(\texttt{R}\,;\,\mathrm{A}_{t}).$
Moreover, the type of
$\mathrm{\textbf{C}}_\texttt{R}(\underline{\mathrm{A}})$ is
$$(|\complement_q(\mathrm{A}_0)|,|\complement_q(\mathrm{A}_1)|,\cdots,|\complement_q(\mathrm{A}_{s-1})|),$$ for
some
$\underline{\mathrm{A}}:=(\mathrm{A}_0,\mathrm{A}_1,\cdots,\mathrm{A}_s)\in\Re_\ell(q,s).$
\end{Theorem}

\begin{Proof} Let $\mathcal{C}$ be an cyclic $\texttt{R}$-linear code of length $\ell.$
From Proposition\,\ref{c(a)}, we have
$\texttt{R}^\ell=\underset{z\in\Sigma_\ell(q)}{\bigoplus}\mathrm{\textbf{C}}_\eta(\texttt{R}\,;\,\{z\})$
and $\mathrm{\textbf{C}}_\eta(\texttt{R}\,;\,\{z\})$'s are free
irreducible cyclic $\texttt{R}$-linear codes. Therefore
$\mathcal{C}=\underset{z\in\Sigma_\ell(q)}{\bigoplus}\mathcal{C}_z,$
where
$\mathcal{C}_z=\mathrm{\textbf{C}}_\eta(\texttt{R}\,;\,\{z\})\cap\mathcal{C}.$
From Proposition\,\ref{irr(a)},
$\mathcal{C}_z=\theta^{t_z}\mathrm{\textbf{C}}_\eta(\texttt{R}\,;\,\{z\}),$
where $t_z\in\{0,1,\cdots,s\}.$ Hence
$$\mathcal{C}=\underset{z\in\Sigma_\ell(q)}{\bigoplus}\theta^{t_z}\mathrm{\textbf{C}}_\eta(\texttt{R}\,;\,\{z\})
=\overset{s-1}{\underset{t=0}{\bigoplus}}\theta^{t}\mathrm{\textbf{C}}_\eta(\texttt{R}\,;\,\mathrm{A}_{t}),$$
where $\mathrm{A}_t=\{z\in\Sigma_\ell\;:\;t_z=t\}.$  Since
$|\Re_\ell(q,s)|=(s+1)^{|\Sigma_\ell(q)|},$ by
Theorem\,\ref{ct-cy}, the uniqueness of
$\underline{\mathrm{A}}:=(\mathrm{A}_0,\mathrm{A}_1,\cdots,\mathrm{A}_s)\in\Re_\ell(q,s)$
such that $\mathcal{C}=\mathrm{\textbf{C}}_\texttt{R}(\mathrm{A})$
is guaranteed.

  Moreover, for
every $t\in\{0,1,\cdots,s-1\},$ the cyclic $\texttt{R}$-linear
code $\mathrm{\textbf{C}}_\eta(\texttt{R}\,;\,\mathrm{A}_{t})$ is
free and
$\texttt{rank}_\texttt{R}(\mathrm{\textbf{C}}_\eta(\texttt{R}\,;\,\mathrm{A}_{t}))=|\complement_q(\mathrm{A}_{t})|.$
Since the direct sum
$\overset{s-1}{\underset{t=0}{\bigoplus}}\theta^{t}\mathrm{\textbf{C}}_\eta(\texttt{R}\,;\,\mathrm{A}_{t})$
gives the type of
$\mathrm{\textbf{C}}_\texttt{R}(\underline{\mathrm{A}}),$ the type
of $\mathrm{\textbf{C}}_\texttt{R}(\underline{\mathrm{A}})$ is
$(k_0,k_1,\cdots,k_{s-1}),$ where
$k_t:=|\complement_q(\mathrm{A}_t)|,$ for every
$t\in\{0,1,\cdots,s-1\}.$
\end{Proof}

\begin{Proposition} Let
$\underline{\mathrm{A}}:=(\mathrm{A}_0,\mathrm{A}_1,\cdots,\mathrm{A}_s)\in\Re_\ell(q,s)$
and $t\in\{0,1,\cdots,s-1\}.$ Then
$\mathrm{\textbf{C}}_\texttt{R}(\underline{\mathrm{A}})^\perp=\mathrm{\textbf{C}}_\texttt{R}(\underline{\mathrm{A}}^{\widetilde{\diamond}}),$
where
$\underline{\mathrm{A}}^{\widetilde{\diamond}}:=(-\mathrm{A}_s,-\mathrm{A}_{s-1},\cdots,-\mathrm{A}_1,-\mathrm{A}_0).$
\end{Proposition}

\begin{Proof} Let $\underline{\mathrm{A}}:=(\mathrm{A}_0,\mathrm{A}_1,\cdots,\mathrm{A}_s)\in\Re_\ell(q,s).$ We have
$\mathrm{\textbf{C}}_\texttt{R}(\underline{\mathrm{A}})^\perp
\supseteq
\bigcap_{u=0}^{s-1}\left(\theta^{s-u}\texttt{R}^\ell+\mathrm{\textbf{C}}_\eta(\texttt{R}\,;\,\mathrm{A}_u^\diamond)\right)$
and
$\theta^{s-t}\mathrm{\textbf{C}}_\eta(\texttt{R}\,;\,-\mathrm{A}_{t})\subseteq\bigcap_{u=0}^{s-1}\left(\theta^{s-u}\texttt{R}^\ell+\mathrm{\textbf{C}}_\eta(\texttt{R}\,;\,\mathrm{A}_u^\diamond)\right),$
for every $t\in\{1,2,\cdots,s\}.$  It follows that
$\mathrm{\textbf{C}}_\texttt{R}(\underline{\mathrm{A}}^{\widetilde{\diamond}})\subseteq\mathrm{\textbf{C}}_\texttt{R}(\underline{\mathrm{A}})^\perp.$
From Proposition\,\ref{dual-type} and Theorem\,\ref{decy},
$\mathrm{\textbf{C}}_\texttt{R}(\underline{\mathrm{A}}^{\widetilde{\diamond}})$
and $
\mathrm{\textbf{C}}_\texttt{R}(\underline{\mathrm{A}})^\perp$ have
the same type, we have
$\mathrm{\textbf{C}}_\texttt{R}(\underline{\mathrm{A}})^\perp=
\mathrm{\textbf{C}}_\texttt{R}(\underline{\mathrm{A}}^{\widetilde{\diamond}}).$
 \end{Proof}

\section{Constacyclic linear codes over a finite chain ring}\label{Sect:5}

  Let $\gamma\in\texttt{R}^\times$ and the multiplicative order of
$\gamma$ is $u.$ We study the structure of contractions of cyclic
$\texttt{R}$-linear codes of length $u\ell.$ In the section, the
usage of the map
\begin{align}
\begin{array}{cccc}
  \wp : & \texttt{R}^\ell & \rightarrow & \texttt{R}^{u\ell} \\
    & \underline{\textbf{c}} & \mapsto & (\gamma^{u-1}\underline{\textbf{c}}\;|\;\gamma^{u-2} \underline{\textbf{c}}\;|\;\cdots\;|\;\gamma\underline{ \textbf{c}},\;|\;\underline{\textbf{c}}),
\end{array}\end{align}
 will be necessary.

\begin{Definition}  Let
$\mathcal{C}$ be an $\texttt{R}$-linear code of length $u\ell$ and
$\mathcal{C}:=\wp(\mathcal{K}).$
\begin{enumerate}
    \item The $\texttt{R}$-linear code
$\mathcal{K}$ is called the \index{contraction}\emph{contraction
of a linear code} of $\mathcal{C}.$
    \item The $\texttt{R}$-linear code
$\mathcal{C}$ is called the \index{cyclic concatenation of a
linear code}\emph{cyclic concatenation} of $\mathcal{K}.$
\end{enumerate}
\end{Definition}

   The contraction of a class of linear cyclic codes over finite
fields have been investigated in \cite{Bie02}. Our contribution is
the generalization of this theory of contraction of cyclic codes
to finite chain rings.

\begin{Lemma}\label{cocy} Let $\gamma\in\texttt{R}^\times$ and $u$ the multiplicative order of
$\gamma.$ Then the map $\wp$ is an $\texttt{R}$-module
monomorphism. Moreover $\wp\circ\tau_\gamma=\tau_1\circ\wp.$
\end{Lemma}

\begin{Proof} It is clear that $\wp$ is an $\texttt{R}$-module
monomorphism. Let
$\underline{\textbf{c}}:=(\textbf{c}_0,\cdots,\textbf{c}_{\ell-1})\in\texttt{R}^\ell$
and $\tau_1$ the cyclic shift on $\{0,1,\cdots,u\ell-1\}.$ We
have:
\begin{eqnarray*}
  \tau_1(\wp(\underline{\textbf{c}})) &=& \tau_1(\cdots\;|\;\gamma^i \textbf{c}_0,\cdots,\gamma^i \textbf{c}_{\ell-1}\;|\;\gamma^{i-1} \textbf{c}_0,\cdots,\gamma^{i-1} \textbf{c}_{\ell-1}\;|\;\cdots); \\
    &=& (\cdots\;|\;\gamma^{i+1} \textbf{c}_{\ell-1},\gamma^i \textbf{c}_0,\cdots,\gamma^i \textbf{c}_{\ell-2}\;|\;\gamma^{i} \textbf{c}_{\ell-1},\gamma^{i-1} \textbf{c}_0,\cdots,\gamma^{i-1} \textbf{c}_{\ell-2}\;|\;\cdots);\\
    &=& \left(\gamma^{u-1}\tau_\gamma(\underline{\textbf{c}})\;|\;\gamma^{u-2} \tau_\gamma(\underline{\textbf{c}})\;|\;\cdots\;|\;\gamma \tau_\gamma(\underline{\textbf{c}})\;|\;\tau_\gamma(\underline{\textbf{c}})\right);\\
    &=& \wp(\tau_\gamma(\underline{\textbf{c}})).
\end{eqnarray*}
Hence $\wp\circ\tau_\gamma=\tau_1\circ\wp.$
\end{Proof}

\begin{Corollary} Let $\gamma\in\texttt{R}^\times$ and $u$ the multiplicative order of
$\gamma$ and  $\mathcal{K}$ be an $\texttt{R}$-linear code of
length $\ell.$ Then $\mathcal{K}$is $\gamma$-constacyclic if and
only if $\wp(\mathcal{K})$ is cyclic $\texttt{R}$-linear code of
length $u\ell.$ Moreover, $\mathcal{K}$ and $\wp(\mathcal{K})$
have the same type.
\end{Corollary}

\begin{Proof} The map $\wp$ is an $\texttt{R}$-module
monomorphism. So $\wp(\mathcal{K})$ is $\texttt{R}$-linear code of
length $u\ell$ and $\mathcal{K},$ $\wp(\mathcal{K})$ have the same
type. From Lemma\,\ref{cocy}, we have
$\wp\circ\tau_\gamma=\tau_1\circ\wp,$ and so $\wp(\mathcal{K})$ is
cyclic.
\end{Proof}

This show how to construct a cyclic $\texttt{R}$-linear code from
a constacyclic $\texttt{R}$-linear code. Now we want to construct
a constacyclic $\texttt{R}$-linear code from a cyclic
$\texttt{R}$-linear code. Let $\mathrm{A}$ be a subset of
$\{0,1,\cdots,u\ell-1\}.$ One denotes
$\mathrm{A}\,(\texttt{mod}\,u):=\{a\,(\texttt{mod}\,u)\;:\;a\in
\mathrm{A}\}.$

\begin{Theorem} Let $u,\ell\in\mathbb{N}$ such that $\texttt{gcd}(u\ell,q)=1.$ Let $\mathrm{A}$ be a subset of $\{0,1,\cdots,u\ell-1\}$ and
$\mathrm{\textbf{C}}_\eta(\texttt{R}\,;\,\mathrm{A})$ be a cyclic
$\texttt{R}$-linear code of length $u\ell.$ Then
$\complement_q(\mathrm{A})\,(\texttt{mod}\,u)=\{\omega\},$ if and
only if
$\mathcal{K}:=\wp^{-1}(\mathrm{\textbf{C}}_\eta(\texttt{R}\,;\,\mathrm{A}))$
is an $\gamma$-constacyclic $\texttt{R}$-linear code of length
$\ell,$ where
$\gamma=\xi^{-\frac{\omega(q^m-1)}{u}\,\texttt{mod}\,u}.$ Moreover
$\mathcal{K}^\perp=\wp^{-1}\left(\mathrm{\textbf{C}}_\eta\left(\texttt{R}\,;\,\mathrm{A}^{\star
u}\right)\right),$ where
$\complement_q(\mathrm{A})\,(\texttt{mod}\,u)=\{\omega\},$ and
$\mathrm{A}^{\star u}:=\left\{a\in\mathrm{A}^\diamond\;:\;a\equiv
-\omega (\texttt{mod}\,u)\right\},$ is an
$\gamma^{-1}$-constacyclic $\texttt{R}$-linear code of length
$\ell.$
\end{Theorem}

\begin{Proof} Let $m$ be the positive integer such that $q^{m}\equiv
1\,(\,\texttt{mod}\,u\ell)$ and $q^{m}\not\equiv
1\,(\,\texttt{mod}\,u\ell).$ Let $\texttt{S}:=\texttt{R}[\xi]$ be
a Galois extension of $\texttt{R}$ of degree $m.$ We set
$\beta:=\xi^{\frac{q^m-1}{u}\,\texttt{mod}\,u},$
$w:=\frac{q^m-1}{u\ell}$ and $\eta:=\xi^w.$ Let
$\mathrm{Z}:=\complement_q(\mathrm{A}),$ where $\mathrm{A}$ is a
subset of $\Sigma_\ell.$ Then
$\mathrm{\textbf{C}}_\eta(\texttt{R}\,;\,\mathrm{Z})=\oplus_{z\in\mathrm{Z}}\mathrm{\textbf{C}}_\eta(\texttt{R}\,;\,\{z\}).$
It is enough to show that
$\mathrm{\textbf{C}}_\eta(\texttt{R}\,;\,\{z\})\subseteq\wp(\texttt{R}^\ell),$
for all $z\in \mathrm{Z}.$ Let $z\in \mathrm{Z},$ we set
$m_z:=|\complement_q(z)|$ and $\zeta:=\eta^{m_z}.$ From Lemma
\,\ref{irr1},
$\mathrm{\textbf{C}}_\eta(\texttt{R}\,;\,\{z\})=\psi_z(\texttt{R}[\xi^{m_z}])=\texttt{Tr}_\texttt{R}^\texttt{S}\left(\emph{ev}_\eta(\texttt{R}[\xi^{m_z}]X^{z})\right).$
Thus for all
$\underline{\textbf{c}}:=(\textbf{c}_0,\cdots,\textbf{c}_{u\ell-1})\in\mathrm{\textbf{C}}_\eta(\texttt{R}\,;\,\{z\}),$
From Lemma\,\ref{irr1}, exist a unique
$\textbf{a}\in\texttt{R}[\xi^{m_z}]$ and  such that
$\underline{\textbf{c}}=\texttt{Tr}_\texttt{R}^\texttt{S}\left(\emph{ev}_\eta(\textbf{a}X^{z})\right).$
Since $\texttt{R}[\xi^{m_z}]$ is the Galois extension of
$\texttt{R}$ of degree $m_z,$ then there exist a unique
$(\textbf{a}_0,\textbf{a}_1,\cdots,\textbf{a}_{m_z-1})$ such that
$\textbf{a}:=\sum\limits_{h=0}^{m_z-1}\textbf{a}_h\xi^{hm_z}\in\texttt{R}[\xi^{m_z}]$
and
$\textbf{c}_t:=\sum\limits_{h=0}^{m_z-1}\textbf{a}_h\texttt{Tr}_\texttt{R}^\texttt{S}\left(\xi^{hm_z+wtz\,\texttt{mod}\,\ell}\right),$
for all $t\in\Sigma_{u\ell}.$ From the euclidian division of
$t\in\{0,1,\cdots,u\ell-1\},$ by $\ell,$ there exists
$(i,j)\in\Sigma_u\times\Sigma_\ell,$ such that $t=i\ell+j.$ We
have:
\begin{eqnarray*}
  \textbf{c}_{i\ell+j}&=& \sum\limits_{h=0}^{m_z-1}\textbf{a}_h\texttt{Tr}_\texttt{R}^\texttt{S}\left(\beta^{zi}\xi^{hm_z+wjz\,\texttt{mod}\,\ell}\right), \text{  since } \beta=\xi^{w\ell\,\texttt{mod}\,u}\text{ and } \beta^{z\,\texttt{mod}\,u}=\beta^{\omega\,\texttt{mod}\,u};\\
    &=& \beta^{\omega i\,\texttt{mod}\,u}\left(\sum\limits_{h=0}^{m_z-1}\textbf{a}_h\texttt{Tr}_\texttt{R}^\texttt{S}\left(\xi^{hm_z+wjz\,\texttt{mod}\,\ell}\right)\right), \text{ since } \beta^{\omega\,\texttt{mod}\,u}\in\texttt{R}; \\
    &=& \beta^{\omega i\,\texttt{mod}\,u}\textbf{x}_j\text{ and } \textbf{x}_j:=\sum\limits_{h=0}^{m_z-1}\textbf{a}_h\texttt{Tr}_\texttt{R}^\texttt{S}\left(\xi^{hm_z+wjz\,\texttt{mod}\,\ell}\right).
\end{eqnarray*}
Thus $\textbf{c}:=(\cdots\;|\;\gamma^i
\textbf{x}_0,\cdots,\gamma^i \textbf{x}_{\ell-1}\;|\;\gamma^{i-1}
\textbf{x}_0,\cdots,\gamma^{i-1} \textbf{x}_{\ell-1}\;|\;\cdots)$
and $\gamma:=\beta^{-\omega}.$ Hence
$\mathrm{\textbf{C}}_\eta(\texttt{R}\,;\,\mathrm{A})\subseteq\wp(\texttt{R}^\ell).$
As $\wp\circ\tau_\gamma=\tau_1\circ\wp,$ it follows that
$\mathcal{K}$ is an $\gamma$-constacyclic $\texttt{R}$-linear code
of length $\ell.$ For sufficiency, it is enough to note that the
above proof is reversible.

As $\mathcal{K}^\perp$ is an $\gamma^{-1}$-constacyclic free
$\texttt{R}$-linear code of rank $\ell-|\mathrm{A}|,$  the cyclic
$\texttt{R}$-linear code which yields $\mathcal{K}^\perp,$ by
contraction must have the definition set $\mathrm{B}$ of size
$|\mathrm{A}|.$
\end{Proof}

\begin{Example} Let $\texttt{R}$ be a finite chain ring of invariants $(q,s)$ where $q=3.$ We take
$\;\ell=28, u=2.$ We set $\mathrm{A}_1:=\complement_q(\{1,7\}),$
$\mathrm{A}_2:=\complement_q(\{1,5,7\}),$ and
$\mathrm{A}_3:=\complement_q(\{1,5,7,11\}).$ We have
$\complement_q(\mathrm{A}_i)\,(\,\texttt{mod}\;2)=\{1\}.$ So we
can set
$\mathcal{K}_i:=\wp^{-1}\left(\mathrm{\textbf{C}}_\eta\left(\texttt{R}\,;\,\mathrm{A}_i\right)\right),$
where $i\in\{1,2,3\}.$ Since $\mathrm{A}_1^{\star 2}=\mathrm{A}_3$
and $\mathrm{A}_2^{\star 2}=\mathrm{A}_2,$ we have
$\mathcal{K}_3=\mathcal{K}_1^\perp$ and $\mathcal{K}_2$ is
self-dual.
\end{Example}

The \emph{Hamming weight} of an $\texttt{R}$-linear code
$\mathcal{C}$ of length $\ell,$  is defined as:
$\texttt{wt}(\mathcal{C}):=\texttt{min}\left\{\texttt{wt}(\underline{\textbf{c}})\,:\,
\textbf{c}\in\mathcal{C}\setminus\{ \textbf{0}\}\right\},$ where
$\texttt{wt}(\underline{\textbf{c}}):=|\{j\in\Sigma_\ell\;:\;\textbf{c}_j\neq
0\}|.$

\begin{Corollary}  Let $u,\ell\in\mathbb{N}$ such that $\texttt{gcd}(u\ell,q)=1.$
Let $\underline{\mathrm{A}}:=(\mathrm{A}_0,\mathrm{A}_1,\cdots,\mathrm{A}_s)\in\Re_{u\ell}(q,s)$ and
$\mathrm{\textbf{C}}_\texttt{R}(\underline{\mathrm{A}})$ be a
cyclic $\texttt{R}$-linear code of length $u\ell$ such that
$\overset{s-1}{\underset{t=0}{\bigcup}}\mathrm{A}_t\,(\,\texttt{mod}\,u)=\{\omega\}.$
Set
$\mathcal{K}:=\wp^{-1}\left(\mathrm{\textbf{C}}_\texttt{R}(\underline{\mathrm{A}})\right).$
Then
\begin{enumerate}
    \item $\mathcal{K}$ is an $\gamma$-constacyclic $\texttt{R}$-linear code of length $\ell,$ where $\gamma=\xi^{-\frac{\omega(q^m-1)}{u}\,\texttt{mod}\,u};$
    \item $\texttt{wt}(\underline{\textbf{c}})=u\cdot\texttt{wt}(\wp^{-1}(\underline{\textbf{c}})),$ for every $\underline{\textbf{c}}\in\mathrm{\textbf{C}}_\texttt{R}(\underline{\mathrm{A}});$
    \item $\mathcal{K}^\perp=\wp^{-1}(\mathrm{\textbf{C}}_\texttt{R}(\underline{\mathrm{A}}^{\star u})),$ where
$\underline{\mathrm{A}}^{\star u}:=(-\mathrm{A}_s^{\star
u},-\mathrm{A}_{s-1},\cdots,-\mathrm{A}_1,-\mathrm{A}_0^{\triangleleft
u})$ with
\begin{itemize}
    \item $\mathrm{A}_s^{\star u}:=\left\{a\in \mathrm{A}_s\;:\;a\equiv -\omega(\,\texttt{mod}\,u)\right\},$
    \item $\mathrm{A}_0^{\triangleleft u}:=\mathrm{A}_0\cup (\mathrm{A}_s\setminus\mathrm{A}_s^{\star u}).$
\end{itemize}
    \end{enumerate}
\end{Corollary}

\begin{Example} Let $\texttt{R}$ be a finite chain ring of invariants $(q,s)$  where $\texttt{J}(\texttt{R})=\texttt{R}\theta,$ $q=3$ and $s=2.$  We take
$\;\ell=10, u=2.$ We set
$\underline{\mathrm{A}}:=(\mathrm{A}_0,\mathrm{A}_1,\mathrm{A}_2),$
where $\mathrm{A}_0:=\complement_q(\{1\}),$
$\mathrm{A}_1:=\complement_q(\{5\}),$ and
$\mathrm{A}_2:=\complement_q(\{0,2,4,10,11\}).$ We have
$\complement_q(\mathrm{A}_0)\,(\,\texttt{mod}\;2)=\complement_q(\mathrm{A}_1)\,(\,\texttt{mod}\;2)=\{1\}.$
So the contraction of the cyclic $\texttt{R}$-linear code
$\textbf{C}_\texttt{R}(\underline{\mathrm{A}})$ of length $20,$ is
the self-dual negacyclic $\texttt{R}$-linear code
$\mathcal{K}:=\wp^{-1}\left(\mathrm{\textbf{C}}_\eta\left(\texttt{R}\,;\,\mathrm{A}_0\right)\right)\oplus\theta\wp^{-1}\left(\mathrm{\textbf{C}}_\eta\left(\texttt{R}\,;\,\mathrm{A}_1\right)\right),$
of length $10.$
\end{Example}

\section{Conclusion}

We have seen that in the case $\texttt{gcd}(\ell,|\texttt{R}|)=1,$
and $\gamma\in\Gamma(\texttt{R})^*,$ the class of
$\gamma$-constacyclic $\texttt{R}$-linear codes of length $\ell,$
is the same as the class of contractions of cyclic
$\texttt{R}$-linear codes
$\mathrm{\textbf{C}}_\texttt{R}(\mathrm{A}_0,\mathrm{A}_1,\cdots,\mathrm{A}_s)$
of length $u\ell,$ where $u$ is the multiplicative order of
$\gamma,$ and each cyclic $\texttt{R}$-linear code
$\mathrm{\textbf{C}}_\texttt{R}(\mathrm{A}_0,\mathrm{A}_1,\cdots,\mathrm{A}_s)$
of this class, satisfies:
$\overset{s-1}{\underset{t=0}{\bigcup}}\mathrm{A}_t\,(\,\texttt{mod}\,u)$
is a singleton.

\end{document}